\documentclass{article}%
\usepackage{makeidx}
\usepackage{amsfonts}
\usepackage{amsmath}
\usepackage{amssymb}
\usepackage{graphicx}%
\providecommand{\U}[1]{\protect\rule{.1in}{.1in}}
\providecommand{\U}[1]{\protect\rule{.1in}{.1in}}
\newtheorem{theorem}{Theorem}

\newtheorem{definition}[theorem]{Definition}

\newtheorem{lemma}[theorem]{Lemma}

\newenvironment{proof}[1][Proof]{\noindent\textbf{#1.} }{\ \rule{0.5em}{0.5em}}
\begin{document}

\title{Invertible coupled KdV and coupled Harry Dym hierarchies}
\author{Maciej B\l aszak\\Faculty of Physics, Division of Mathematical Physics, A. Mickiewicz University\\Umultowska 85, 61-614 Pozna\'{n}, Poland\\blaszakm@amu.edu.pl
\and Krzysztof Marciniak\\Department of Science and Technology \\Campus Norrk\"{o}ping, Link\"{o}ping University\\601-74 Norrk\"{o}ping, Sweden\\krzma@itn.liu.se}
\maketitle

\begin{abstract}
In this paper we discuss the conditions under which the coupled KdV and
coupled Harry Dym hierarchies possess inverse (negative) parts. We further
investigate the structure of nonlocal parts of tensor invariants of these
hierarchies, in particular, the nonlocal terms of vector fields, conserved
one-forms, recursion operators, Poisson and symplectic operators. We show that
the invertible cKdV hierarchies possess Poisson structures that are at most
weakly nonlocal while coupled Harry Dym hierarchies have Poisson structures
with nonlocalities of the third order.

\end{abstract}


\section{Introduction}

The energy-dependent Schr\"{o}dinger spectral problem has been introduced by
Jaulent and Miodek in \cite{JM} in the two-field case. It has been generalized
to an arbitrary number of components by Mart\'{\i}nez Alonso \cite{MA} who
also presented its multi-Hamiltonian structure and gave the problem its
present name. Antonowicz and Fordy have further investigated this problem in a
series of papers (\cite{AF1},\cite{AF2},\cite{AF3},\cite{AF4},\cite{AF5}).
They demonstrated that this spectral problem leads to two families of coupled
(multicomponent) soliton hierarchies: the coupled KdV (cKdV) and the coupled
Harry-Dym (cHD) hierarchies. In their approach one simultaneously obtains the
evolution equations of the hierarchy together with the set of independent
closed one-forms (variational derivatives of Hamiltonians), which can be
obtained with the help of a recursion relation, solvable under additional
conditions. The specification of these conditions fixes the type of the
hierarchy (either cKdV or cHD).

In this paper we complete their work by finding the conditions under which
both types of hierarchies have also negative parts (i.e., when the recursion
operator is explicitly invertible) and present a first few flows of these
negative hierarchies. Further, using the results from \cite{EOR} and \cite{S},
we investigate the structure of nonlocal parts of all Hamiltonian structures
associated with both types of hierarchies: we show that all Hamiltonian and
symplectic structures of the cKdV-type hierarchies are at most weakly nonlocal
while the cHD-type hierarchies have Hamiltonian and symplectic structures that
are nonlocal up to third order, i.e., they have terms up to $\partial^{-3}$ in
their nonlocal parts. We also present compact formulas for the highest
nonlocalities of all these quantities.

\section{The spectral problem}

We start by recalling the basic facts about coupled (multicomponent) KdV and
Harry Dym hierarchies following the papers \cite{AF2}, \cite{AF3} and
\cite{AF5} (see also \cite{M}). Let us consider the Schr\"{o}dinger equation
\begin{equation}
\mathcal{L}\Psi\equiv(\varepsilon\partial^{2}+u)\Psi=0\label{1}%
\end{equation}
where $\Psi=\Psi(x,t)$ and $u=u(x,t)$ are smooth functions of $x$ and $t$ and
where $\partial=\frac{\partial}{\partial x}$, together with the auxiliary
linear problem of the form
\begin{equation}
\Psi_{t}=\left(  \tfrac{1}{2}P\partial+Q\right)  \Psi.\label{2}%
\end{equation}
The functions $P$ and $Q$ will be specified below.

The system (\ref{1})-(\ref{2}) is compatible, i.e., has nontrivial solutions
for $\Psi$ provided that $\left(  \mathcal{L}\Psi\right)  _{t}=0$ which yields%
\[
\frac{1}{2}\varepsilon(P_{xx}+4Q_{x})\Psi_{x}+\left(  \varepsilon
Q_{xx}-uP_{x}-\frac{1}{2}Pu_{x}+u_{t}\right)  \Psi=0.
\]
Here and below the subscripts $x$ and $t$ denote the derivatives with respect
to $x$ and $t$ respectively. Thus, the equations%
\[
P_{xx}+4Q_{x}=0,\qquad\varepsilon Q_{xx}-uP_{x}-\frac{1}{2}Pu_{x}+u_{t}=0
\]
constitute a set of sufficient conditions for the existence of a common
solution $\Psi(x,t)$ to (\ref{1})-(\ref{2}). The first equation can be
replaced by $Q=-\frac{1}{4}P_{x}$, and then the second one reads%
\begin{equation}
u_{t}=JP\equiv\left(  \frac{1}{4}\varepsilon\partial^{3}+\frac{1}{2}\left(
u\partial+\partial u\right)  \right)  P,\label{evol}%
\end{equation}
where $J$ is a third-order differential operator given above. Assume now that
both $\varepsilon$ and $u$ are polynomial functions (of the same degree) of a
new parameter $\lambda$ so that%
\begin{equation}
u=\sum_{k=0}^{N}u_{k}(x,t)\lambda^{k},\quad\varepsilon=\sum_{k=0}%
^{N}\varepsilon_{k}\lambda^{k},\label{uek}%
\end{equation}
where $\varepsilon_{k}$ are $N+1$ arbitrary (so far) real constants. Then $J $
is polynomial as well: $J=\sum_{k=0}^{N}J_{k}\lambda^{k}$ with%
\begin{equation}
J_{k}=\frac{1}{4}\varepsilon_{k}\partial^{3}+\frac{1}{2}\left(  u_{k}%
\partial+\partial u_{k}\right)  ,\text{ \ }k=0,\ldots,N.\label{Jk}%
\end{equation}
Assume further that also $P$ (and thus $Q$) depends polynomially on $\lambda$
so that%
\begin{equation}
P=\sum_{k=0}^{m}P_{k}\lambda^{m-k},\quad m\in\mathbf{N.}\label{P}%
\end{equation}
It is not the only possibility (for example, in \cite{M} we have considered a
rational dependence of $P$ on $\lambda$ which leads to hierarchies with
sources) but here we restrict ourselves to (\ref{P}). Plugging (\ref{uek}) and
(\ref{P}) into (\ref{evol}) we obtain, for any fixed $m\in\mathbf{N}$, the
following $(N+1)$-component system of evolution equations%
\begin{equation}
u_{r,t_{m}}=J_{0}P_{m-r}+\cdots+J_{r}P_{m}~,~r=0,\ldots,N\label{ewolucja}%
\end{equation}
(with $P_{i}=0$ for $i<0$) and the following recursion on $P_{i}$
\begin{equation}
J_{0}P_{k-N}+J_{1}P_{k-N+1}+\cdots+J_{N}P_{k}=0\text{, \ \ }k=0,\ldots
,m-1.\label{rek}%
\end{equation}
Note that the natural parameters $N$ and $m$ are independent; $N$ will be the
number of fields in the hierarchies that originate in this scheme while $m$
enumerates the flows within a chosen hierarchy. The recursion (\ref{rek})
starts at $k=0$ with the equation $J_{N}P_{0}=0$ and it can be effectively
solved either when $u_{N}=0$ or when $\varepsilon_{N}=0$. In the case of
$u_{N}=0$ we have $J_{N}=\frac{1}{4}\varepsilon\partial^{3}$ which leads to
$P_{0}$ and thus $P_{r}$ depending explicitly on $x$ which we do not consider
here. We must therefore assume (and will stick to this throughout the whole
article) that $\varepsilon_{N}=0$ so that $J_{N}=\frac{1}{2}\left(
u_{N}\partial+\partial u_{N}\right)  =u_{N}^{1/2}\partial u_{N}^{1/2}$ which
implies that $J_{N}$ is invertible with the inverse $J_{N}^{-1}=u_{N}%
^{-1/2}\partial^{-1}u_{N}^{-1/2},$ where $\partial^{-1}$ is a formal inverse
of $\partial$ ($\partial\partial^{-1}=\partial^{-1}\partial=1$). The recursion
(\ref{rek}) allows now to compute $P_{0}$, $P_{1}$ and so on up to $P_{m-1}$
as differential functions of $u$ with $P_{m}$ undetermined since $P_{m}$ does
not enter the formulas (\ref{rek}). In the construction given in \cite{AF5}
one can demonstrate that in order to embed (\ref{ewolucja}) into an infinite
hierarchy with well-defined and explicitly computable $P_{m}$ for all
$m\in\mathbf{N}$ one must make a reduction of the system (\ref{ewolucja}) to
an $N$-component system by assuming either that $u_{N}$ is constant (to make
the coefficients in the hierarchy as simple as possible, the most convenient
choice is to set $u_{N}=-1$) or that $u_{0}=-a^{2}$ ($a$ is a constant). The
first choice leads to the coupled KdV hierarchy, the second to the coupled
Harry Dym (cHD) hierarchy. But, in fact, under the additional assumption
($\varepsilon_{0}=0$ in the KdV case and $u_{0}=-a^{2}=0$ in the HD case) one
can invert the operator $J_{0}$ which makes it possible to construct both KdV
and HD negative hierarchies even in the multicomponent (coupled) case. The
existence of the inverse cKdV hierarchy was mentioned in \cite{AF5}, but never
treated in detail. The one-field case of the inverse HD hierarchy was
considered in \cite{Brunelli}.

There is actually a third possibility of obtaining a reversible hierarchy in
this scheme: by putting $\varepsilon_{N}\neq0$ but $u_{N}=0$ and
$\varepsilon_{0}=0$ we also arrive at an invertible hierarchy but it can be
shown that this hierarchy is just a reparametrization of the invertible
coupled Harry Dym hierarchy described below.

\section{Some algebraic structures and their nonlocal parts}

In order to state our results we need to introduce some algebraic objects and
to discuss their nonlocal parts. Following Antonowicz and Fordy, denote by
$B_{0}$ the Hamiltonian operator%
\[
B_{0}=\left(
\begin{array}
[c]{cccc}%
~~-J_{1} & ~-J_{2}~ & ~\cdots~ & -J_{N}\\
~~-J_{2} & ~~ & -J_{N} & \\
\vdots~~ & \cdots & ~~ & \\
-J_{N} &  &  &
\end{array}
\right)
\]
and by $R$ the following nonlocal operator of $(1,1)$-type%
\begin{equation}
R=\left(
\begin{array}
[c]{c|c}%
\begin{array}
[c]{ccc}%
0 & \ldots & 0
\end{array}
& -J_{0}J_{N}^{-1}\\\hline
I_{N-1} &
\begin{array}
[c]{c}%
-J_{1}J_{N}^{-1}\\
\vdots\\
-J_{N-1}J_{N}^{-1}%
\end{array}
\end{array}
\right) \label{R}%
\end{equation}
where, as in the previous section,
\begin{equation}
J_{i}=\frac{1}{4}\varepsilon_{i}\partial^{3}+\frac{1}{2}\left(  u_{i}%
\partial+\partial u_{i}\right)  ,\text{ \ }i=0,\ldots,N\label{Ji}%
\end{equation}
with $u_{i}=u_{i}(x,t)$, $i=0,\ldots,N$, and with $\varepsilon_{N}=0$ so that
$J_{N}=\frac{1}{2}\left(  u_{N}\partial+\partial u_{N}\right)  =$ $u_{N}%
^{1/2}\partial u_{N}^{1/2}$ is invertible. Here and below $I_{N-1}$ stands for
the $(N-1)\times(N-1)$ unit matrix. Note also that (\ref{R}) implies that the
operator $R^{\dagger}$ adjoint to $R$ has the form%
\begin{equation}
R^{\dagger}=\left(
\begin{array}
[c]{c|c}%
\begin{array}
[c]{c}%
0\\
\vdots\\
0
\end{array}
& I_{N-1}\\\hline
-J_{N}^{-1}J_{0} &
\begin{array}
[c]{ccc}%
-J_{N}^{-1}J_{1} & \cdots & -J_{N}^{-1}J_{N-1}%
\end{array}
\end{array}
\right) \label{Rad}%
\end{equation}
(since $J_{i}^{\dagger}=-J_{i}$ for all $i$). It can be proved that the
operator $R$ is a hereditary recursion operator \cite{FF} meaning that its
Nijenhuis torsion vanishes. Let us now assume that the operator $J_{0}$ is
also invertible (this can be achieved either by putting $\varepsilon_{0}=0$ or
by putting $u_{0}=0$, as discussed below). Then $R$ is invertible with%
\[
R^{-1}=\left(
\begin{array}
[c]{c|c}%
\begin{array}
[c]{c}%
-J_{1}J_{0}^{-1}\\
\vdots\\
-J_{N-1}J_{0}^{-1}%
\end{array}
& I_{N-1}\\\hline
-J_{N}J_{0}^{-1} &
\begin{array}
[c]{ccc}%
0 & \cdots & 0
\end{array}
\end{array}
\right)
\]
which implies%
\[
\left(  R^{-1}\right)  ^{\dagger}=\left(
\begin{array}
[c]{c|c}%
\begin{array}
[c]{ccc}%
-J_{0}^{-1}J_{1} & \cdots & -J_{0}^{-1}J_{N-1}%
\end{array}
& -J_{0}^{-1}J_{N}\\\hline
I_{N-1} &
\begin{array}
[c]{c}%
0\\
\vdots\\
0
\end{array}
\end{array}
\right)  .
\]
Now, since $R$ is invertible, and due to the fact that $R$ is hereditary, the
infinite family of operators%
\begin{equation}
B_{s}=R^{s}B_{0},\text{ \ }s\in\mathbf{Z}\label{Br}%
\end{equation}
is a family of compatible Poisson operators. The operators $B_{s}$ are purely
local for $s=0,\ldots,N$; explicitly, they have the form
\begin{equation}
B_{s}=\left(
\begin{array}
[c]{cccc|cccc}%
~ & ~ & ~ & ~J_{0}~ & ~ &  & ~ & ~\\
~ & ~ & ~J_{0}~ & ~J_{1}~ & ~ & ~ & 0 & ~\\
~ & \cdots\  & ~ & \vdots & ~ & ~ & ~ & ~\\
~J_{0}~ & ~J_{1}~ & \cdots & ~J_{s-1}~ & ~ & ~ & ~ & ~\\\hline
~ & ~ & ~ & ~ & -J_{s+1} & \cdots & -J_{N-1} & -J_{N}\\
~ & ~ & ~ & ~ & \vdots & ~ & \cdots & ~\\
~ & 0 & ~ & ~ & -J_{N-1} & -J_{N} & ~ & ~\\
~ & ~ & ~ & ~ & -J_{N} & ~ & ~ & ~
\end{array}
\right)  ,\text{ \ }s=0,\ldots,N\label{Bry}%
\end{equation}
so that $B_{0}$ is as above and
\[
B_{N}=\left(
\begin{array}
[c]{cccc}
&  & ~ & J_{0}\\
&  & J_{0} & J_{1}\\
~~ & \ldots & ~~ & \vdots\\
J_{0} & J_{1} & \ldots & J_{N-1}%
\end{array}
\right)  .
\]
As we will show below, all the other $B_{s}$ are nonlocal.

Since $B_{0}$ is invertible and Poisson, its inverse $\Omega_{0}=B_{0}^{-1}$
is a closed two-form. We can therefore also define a family of two-forms%
\begin{equation}
\Omega_{s}=\left(  R^{\dagger}\right)  ^{s}\Omega_{0}\text{, }s\in
\mathbf{Z}.\label{omegi}%
\end{equation}
Moreover, according to the theorem given below, $\Omega_{s}^{-1}=B_{-s}$ so
that all $\Omega_{s}$ constitute a family of closed two-forms.

\begin{theorem}
For the families $B_{s}$ and $\Omega_{s}$ defined above we have%
\[
B_{s}\Omega_{-s}=I_{N}\text{ for }s\in\mathbf{Z}.
\]

\end{theorem}

\begin{proof}
Since $B_{s}=R^{s}B_{0}$ we have $R^{s}=B_{s}\Omega_{0}$ so that $\left(
R^{\dagger}\right)  ^{s}=\Omega_{0}B_{s}$ (note that $\Omega_{0}^{\dagger
}=-\Omega_{0}$ and $B_{s}^{\dagger}=-B_{s}$) and thus $\Omega_{s}=\left(
R^{\dagger}\right)  ^{s}\Omega_{0}=\Omega_{0}B_{s}\Omega_{0}$. We have then%
\[
B_{s}\Omega_{-s}=R^{s}B_{0}\Omega_{0}B_{-s}\Omega_{0}=R^{s}B_{-s}\Omega
_{0}=R^{s}R^{-s}B_{0}\Omega_{0}=I_{N}%
\]

\end{proof}

Let us now choose two one-forms $\ \gamma_{0}$ and $\gamma_{-1}$ (they will be
"starting points" of our hierarchies later on) so that%
\begin{equation}
\gamma_{0}=(0,\ldots,0,\varphi)^{T}\text{ with \ }\varphi\in\text{Ker}%
J_{N}\text{, \ \ \ }\gamma_{-1}=(\psi,0,\ldots,0)^{T}\text{ with \ }\psi
\in\text{Ker}J_{0}.\label{start}%
\end{equation}
We have then
\[
\gamma_{0}\in\text{Ker}\left(  (R^{\dagger})^{-1}\right)  \text{, \ }%
\gamma_{-1}\in\text{Ker}(R^{\dagger})\text{\ }%
\]
so that
\begin{equation}
\left(  R^{\dagger}\right)  ^{-1}\gamma_{0}=0\text{, \ \ }R^{\dagger}%
\gamma_{-1}=0.\label{zerowanie}%
\end{equation}
Moreover, it is easy to check that $\gamma_{0}$ belongs to the kernel of each
of the operators $B_{0},\ldots,B_{N-1}$%
\begin{equation}
\gamma_{0}\in\text{Ker}B_{s}\text{ \ for all }s=0,\ldots,N-1.\label{wlasnosc}%
\end{equation}
We define now two families of one-forms%
\begin{equation}
\gamma_{s}=\left(  R^{\dagger}\right)  ^{s}\gamma_{0},\text{ \ \ \ }%
\gamma_{-s}=\left(  \left(  R^{-1}\right)  ^{\dagger}\right)  ^{s-1}%
\gamma_{-1},\text{ \ \ }s=1,2,\ldots\label{gammy}%
\end{equation}
and it follows from (\ref{zerowanie}) and (\ref{gammy}) that for any two
non-negative $s$ and $p$ we have
\begin{equation}
\left(  R^{\dagger}\right)  ^{s}\gamma_{-p}=\left\{
\begin{array}
[c]{c}%
\gamma_{-p+s}\text{ for }s<p\\
0\text{ for }s\geq p
\end{array}
\right.  \text{ \ and \ }\left(  \left(  R^{-1}\right)  ^{\dagger}\right)
^{s}\gamma_{p}=\left\{
\begin{array}
[c]{c}%
\gamma_{p-s}\text{ for }s\leq p\\
0\text{ for }s>p
\end{array}
\right. \label{uwaga}%
\end{equation}
We can now show exactly which $\gamma_{i}$ belong to Ker$B_{s}$ for any given
$s\in Z$.

\begin{lemma}
\label{cas}From the property (\ref{wlasnosc}) it follows that%
\begin{align}
\gamma_{-s},\ldots,\gamma_{-s+N-1}  & \in Ker(B_{s})\text{ for }0\leq s\leq
N,\nonumber\\
\gamma_{-N-s},\ldots,\gamma_{-1}  & \in Ker(B_{N+s})\text{ for }%
s>0,\label{jadraB}\\
\gamma_{0},\ldots,\gamma_{s+N-1}  & \in Ker(B_{-s})\text{ for }s>0.\nonumber
\end{align}

\end{lemma}

\begin{proof}
The condition (\ref{wlasnosc}) means that $B_{s}\gamma_{0}=0$ for
$s=0,\ldots,N-1$. But $B_{s}\gamma_{0}=R^{s}B_{0}\gamma_{0}=B_{0}\left(
R^{\dagger}\right)  ^{s}\gamma_{0}$ $=B_{0}\gamma_{s}$ so that%
\begin{equation}
\gamma_{0},\ldots,\gamma_{N-1}\in Ker(B_{0}).\label{jadro0}%
\end{equation}
Further, if $0\leq s\leq N$ and if $p\geq0$ then $B_{s}\gamma_{p}=R^{s}%
B_{0}\gamma_{p}=B_{0}\left(  R^{\dagger}\right)  ^{s}\gamma_{p}=B_{0}%
\gamma_{p+s}=0$ as soon as $p+s\leq N-1$ (by (\ref{jadro0})), i.e. as soon as
$p\leq N-s-1$. If we still have $0\leq s\leq N$ but $p<0$ then by the same
calculation $B_{s}\gamma_{p}=B_{0}\left(  R^{\dagger}\right)  ^{s}\gamma
_{p}=0$ whenever $s\geq-p$ or $p\geq-s$ (due to the first formula in
(\ref{uwaga})). This proves the first formula in (\ref{jadraB}). The argument
for $B_{N+s}$ is similar. Finally, for $s>0,$ $B_{-s}\gamma_{p}=B_{0}\left(
R^{\dagger}\right)  ^{-s}\gamma_{p}$ is equal to $0$ $\ $for $0\leq-s+p\leq
N-1$ (by (\ref{jadro0})) and it is also equal to $0$ for $p=0,\ldots
,s-1\,\ $(by the second formula in (\ref{uwaga})). This yields the third
formula in (\ref{jadraB}).
\end{proof}

Thus, the purely local Poisson tensors $B_{s},\ s=0,\dots,N$ have $N$ Casimir
one-forms $\gamma_{i}$ each, while all the other Poisson tensors $B_{s}%
,$(which will be shown to be nonlocal) have $N+s$ Casimir one-forms
$\gamma_{i}$ each.

In this article we will study the leading nonlocal terms of invariant tensor
objects (vector fields, closed one-forms, closed two-forms, Poisson tensors
and recursion operators) for the invertible coupled KdV and coupled Harry Dym
hierarchies. In order to do this, we have to establish some facts about
nonlocal linear differential operators. Throughout the whole article we will
deal with linear matrix pseudo-differential operators of the form%
\begin{equation}
\Phi=\Phi_{>-n}+\sum_{\alpha=1}^{p}W_{\alpha}\partial^{-n}\varphi_{\alpha}%
^{T},\label{op}%
\end{equation}
where $n\in\mathbf{Z}_{+}$, $W_{\alpha}$ and $\varphi_{\alpha}$ are some
column matrices with entries being some functions of $x$ and where $\Phi
_{>-n}$ denotes the part of the operator $\Phi$ that contains all the local
terms and all the nonlocal terms up to $\partial^{-n+1}$. We will also call
the number $n$ the \emph{order of nonlocality} of the operator $\Phi$. By
$\Phi_{-n}$ we mean the highest nonlocal term of the operator $\Phi$ in
(\ref{op}), i.e.,%
\[
\Phi_{-n}=\sum_{\alpha=1}^{p}W_{\alpha}\partial^{-n}\varphi_{\alpha}^{T}.
\]
We will now state an important theorem that generalizes formula (7) from
\cite{S}.

\begin{theorem}
\label{zlozenie}Consider two linear matrix nonlocal differential operators of
the form%
\begin{equation}
\Phi=\Phi_{>-n}+\sum_{\alpha=1}^{p}W_{\alpha}\partial^{-n}\varphi_{\alpha}%
^{T}\text{ \ and \ }\widetilde{\Phi}=\widetilde{\Phi}_{>-m}+\sum_{\alpha
=1}^{\widetilde{p}}\widetilde{W}_{\alpha}\partial^{-m}\widetilde{\varphi
}_{\alpha}^{T}\label{opy}%
\end{equation}
where $m$ and $n$ are some natural numbers and where $W_{\alpha}%
,\varphi_{\alpha},\widetilde{W}_{\alpha},\widetilde{\varphi}_{\alpha}$ are
some column matrices with entries being some functions of $x$. Assume that all
the products $\varphi_{\alpha}\widetilde{W}_{\beta}$ in (\ref{opy}) are not
constant. Then, the product $\Phi$\ $\widetilde{\Phi}$ has nonlocality of
order $\max\{m,n\}$ and its highest nonlocal term is in the case of $n>m$
given by%
\[
\left(  \Phi\ \widetilde{\Phi}\right)  _{-n}=\sum_{\alpha=1}^{p}W_{\alpha
}\partial^{-n}\left[  \widetilde{\Phi}^{\dagger}\left(  \varphi_{\alpha
}\right)  \right]  ^{T}%
\]
in case of $n=m$ it is given by%
\[
\left(  \Phi\ \widetilde{\Phi}\right)  _{-n}=\sum_{\alpha=1}^{p}W_{\alpha
}\partial^{-n}\left[  \widetilde{\Phi}^{\dagger}\left(  \varphi_{\alpha
}\right)  \right]  ^{T}+\sum_{\alpha=1}^{\widetilde{p}}\Phi\left(
\widetilde{W}_{\alpha}\right)  \partial^{-n}\widetilde{\varphi}_{\alpha}^{T}%
\]
and in the case of $n<m$ it is given by%
\[
\left(  \Phi\ \widetilde{\Phi}\right)  _{-m}=\sum_{\alpha=1}^{\widetilde{p}%
}\Phi\left(  \widetilde{W}_{\alpha}\right)  \partial^{-m}\widetilde{\varphi
}_{\alpha}^{T}.
\]

\end{theorem}

The proof of this theorem consists in a rather straightforward computation
based on the following lemma:

\begin{lemma}
If $f$ is a non-constant function of $x$ then
\[
\partial^{-n}f\partial^{-m}=\left\{
\begin{array}
[c]{c}%
\left(  \partial^{-n}f\right)  \partial^{-m}+\text{lower,~ for }n<m\\
\left(  \partial^{-n}f\right)  \partial^{-n}+(-1)^{n}\partial^{-n}\left(
\partial^{-n}f\right)  +\text{ lower, \ for }n=m\\
(-1)^{m}\partial^{-n}\left(  \partial^{-m}f\right)  +\text{lower, for }n>m
\end{array}
\right.
\]
where the word "lower" means nonlocal terms of lower order of nonlocality.
\end{lemma}

One proves this lemma by repeated integration by parts. Now, the repeated use
of Theorem \ref{zlozenie} leads to

\begin{theorem}
\label{potega}Suppose that no $\varphi_{\alpha}W_{\beta}$ in the operator
$\Phi$ given by (\ref{op}) is constant. Then the $s$-th power $\Phi^{s}$
($s\in\mathbf{Z}_{+}$) of the operator $\Phi$ has nonlocality of the same
order as $\Phi$ and the following formula is valid%
\begin{equation}
(\Phi^{s})_{-n}=%
{\displaystyle\sum\limits_{j=0}^{s-1}}
{\displaystyle\sum\limits_{\alpha=1}^{p}}
\Phi^{j}\left(  W_{\alpha}\right)  \partial^{-n}\left[  \left(  \Phi^{\dagger
}\right)  ^{s-1-j}\left(  \varphi_{\alpha}\right)  \right]  ^{T}.\label{fin}%
\end{equation}

\end{theorem}

One proves this theorem by induction. A special case of this theorem (for
$n=1$) can be found in \cite{S}, but it includes the erronous coefficients
$\binom{s-1}{j}$ that are not present in the correct version of the formula.

\section{Invertible coupled KdV hierarchy}

We are now in position to define our invertible hierarchies. The invariant
one-forms of both hierarchies will be generated by the formulas (\ref{gammy}).
Let us start with the invertible coupled KdV hierarchy. This hierarchy
originates from (\ref{ewolucja}) when we set $\varepsilon_{0}=\varepsilon
_{N}=0$, $u_{N}=-1$ and use the powers of $R^{\dagger}$ and $(R^{\dagger
})^{-1}$ as described above.

\begin{definition}
The invertible $N$-component coupled KdV (invertible cKdV) hierarchy is the
family of flows (i.e. systems of evolutionary PDE's)%
\begin{align}
\frac{d}{dt_{-s}}u  & =K_{-s}\equiv B_{r}\gamma_{-r-s}\text{, \ }%
s=1,2,\ldots\text{ and for all }r>-s\nonumber\\
& \label{IcKdV}\\
\frac{d}{dt_{s}}u  & =K_{s}\equiv B_{r}\gamma_{-r+s+N}\text{, \ }%
s=0,1,2,\ldots\text{ and for all }r\leq s+N\nonumber
\end{align}
where $u=(u_{0},\ldots,u_{N-1})^{T}$, $u_{i}=u_{i}(x,t)$, $B_{r}$ are
Hamiltonian operators defined in (\ref{Br}) with $\varepsilon_{0}%
=\varepsilon_{N}=0$, $u_{N}=-1$, where $\gamma_{s}$ are one-forms given by
(\ref{gammy})%
\begin{equation}
\gamma_{s}=\left(  R^{\dagger}\right)  ^{s}\gamma_{0},\text{ \ \ \ }%
\gamma_{-s}=\left(  \left(  R^{-1}\right)  ^{\dagger}\right)  ^{s-1}%
\gamma_{-1},\text{ \ \ }s=1,2,\ldots\label{gr}%
\end{equation}
with $\gamma_{0}$ and $\gamma_{-1}$ given by
\begin{equation}
\gamma_{0}=(0,\ldots,0,2)^{T},\text{ \ \ \ }\gamma_{-1}=(u_{0}^{-1/2}%
,0,\ldots,0)^{T}.\label{wybgkdv}%
\end{equation}

\end{definition}

\bigskip Note that $\gamma_{0}$ and $\gamma_{-1}$ defined by (\ref{wybgkdv})
do satisfy (\ref{start}) so that $\gamma_{0}\in$Ker$\left(  (R^{\dagger}%
)^{-1}\right)  $ while $\gamma_{-1}\in\in$Ker$(R^{\dagger}).$ Explicitly, for
the invertible cKdV hierarchy we have%
\begin{align*}
J_{0}  & =\frac{1}{2}\left(  u_{0}\partial+\partial u_{0}\right)  ,\\
J_{i}  & =\frac{1}{4}\varepsilon_{i}\partial^{3}+\frac{1}{2}\left(
u_{i}\partial+\partial u_{i}\right)  ,\text{ \ }i=1,\ldots,N-1,\\
J_{N}  & =-\partial.
\end{align*}
\bigskip Note also that (\ref{IcKdV}) implies
\begin{equation}
B_{r}\gamma_{-s}=K_{r-s},\text{ for all }r<s,\text{ \ \ \ }B_{r}\gamma
_{s}=K_{s+r-N}\text{ \ for all }r\geq N-s\label{26}%
\end{equation}
and also
\begin{align}
K_{s}  & =R^{s}K_{0},\text{ }s=1,2,\ldots\text{ with }K_{0}=B_{0}\gamma
_{0},\label{27}\\
K_{-s}  & =\left(  R^{-1}\right)  ^{s-1}K_{-1},~\,s=1,2,\ldots\,\ \text{with
}K_{-1}=B_{0}\gamma_{-1},\label{28}%
\end{align}
so that by the definition above all the vector fields $K_{s}$ have infinitely
many equivalent Hamiltonian representations. If we denote by $\left(
K_{s}\right)  _{i}$ the $i$-th component of the vector field $K_{s}$ then the
first few members of this double-infinite hierarchy are

\bigskip%
\begin{align*}
\left(  K_{-1}\right)  _{i}  & =-\frac{1}{4}\varepsilon_{i}\left(
u_{0}^{-1/2}\right)  _{xxx}-u_{i}\left(  u_{0}^{-1/2}\right)  _{x}-\frac{1}%
{2}u_{ix}u_{0}^{-1/2}\\
\left(  K_{0}\right)  _{i}  & =u_{i-1,x}\\
\left(  K_{1}\right)  _{i}  & =u_{i-2,x}+\frac{1}{4}\varepsilon_{i-1}%
u_{N-1,xxx}+u_{N-1,x}u_{i-1}+\frac{1}{2}u_{N-1}u_{i-1,x}%
\end{align*}
while the first few one-forms $\gamma_{s}$ are given by%
\begin{align*}
\gamma_{-2}  & =\left(  \frac{1}{4}\varepsilon_{1}u_{0}^{-1/2}\left(
u_{0}^{-1/2}\left(  u_{0}^{-1/2}\right)  _{xx}-\frac{1}{2}\left(  u_{0}%
^{-1/2}\right)  _{x}^{2}\right)  +\frac{1}{2}u_{0}^{-3/2},u_{0}^{-1/2}%
,0,\ldots,0\right)  ^{T}\\
\gamma_{-1}  & =(u_{0}^{-1/2},0,\ldots,0)^{T}\\
\gamma_{0}  & =(0,\ldots,0,2)^{T}\\
\gamma_{1}  & =(0,\ldots,0,2,u_{N-1})^{T}\\
\gamma_{2}  & =\left(  0,\ldots,0,2,u_{N-1},u_{N-2}+\frac{1}{4}\varepsilon
_{N-1}u_{N-1,xx}+\frac{3}{4}u_{N-1}^{2}\right)  ^{T}.
\end{align*}
Moreover, due to the hereditary property of $R$, we have that $[K_{i}%
,K\,_{j}]=0$ for all $i,j\in\mathbf{Z}$, $d\gamma_{i}=0$ ($\gamma_{i}$ are all
exact one-forms), $L_{K_{i}}R=0$ for all $i\in\mathbf{Z}$, $L_{K_{i}}%
\gamma_{j}=0$ for all $i,j\in\mathbf{Z}$.

Now, it is possible to show that
\[
\gamma_{s}=(P_{s-N+1},\ldots,P_{s})^{T}\,,\text{ \ \ }s=1,2,...\text{ with
}P_{\alpha}=0\text{ for }\alpha<0
\]
with $P_{s}$ defined in (\ref{P}) being exactly the same as those originally
given in the papers of Antonowicz and Fordy. That means that our hierarchy is
indeed a negative extension of the cKdV hierarchy considered in \cite{AF2}.
Note also that this hierarchy depends on $N-1$ free parameters $\varepsilon
_{1},\ldots,\varepsilon_{N-1}$ (since both $\varepsilon_{0}$ and
$\varepsilon_{N}$ are zero, contrary to the case considered in \cite{AF2}
where it depended on $N$ parameters $\varepsilon_{1},\ldots,\varepsilon_{N}$).
Let us also remark that in the one-field case (i.e., when $N=1$ so that only
$u_{0}$ evolves) our hierarchy becomes equivalent to the dispersionless KdV hierarchy.

Let us now investigate the nature of nonlocalities arising from the invertible
cKdV hierarchy. We start by observing that if we split the operators $R$ and
$R^{-1}$ into their positive (purely local) and negative (purely nonlocal)
parts then we obtain
\begin{align}
R  & =R_{+}+R_{-}=R_{+}+\frac{1}{4}K_{0}\partial^{-1}\gamma_{0}^{T}%
\label{split}\\
R^{-1}  & =\left(  R^{-1}\right)  _{+}+\left(  R^{-1}\right)  _{-}=\left(
R^{-1}\right)  _{+}+K_{-1}\partial^{-1}\gamma_{-1}^{T}\nonumber
\end{align}
($R_{+}=R_{>-1}$ and $R_{-}=R_{-1}$ in the notation from the previous chapter)
so that both $R$ and $R^{-1}$ have nonlocalities of order $1$.

\begin{theorem}
The vector fields $K_{s}$ are local for all $s\in\mathbf{Z}$.
\end{theorem}

\begin{proof}
We will prove this statement inductively with respect to $s$ and separately
for the positive and for the negative part of the hierarchy. First consider
the positive hierarchy. The vector $K_{0}$ is local. Assume now that $K_{s}$
is local. Then, due to (\ref{split})%
\[
K_{s+1}=R(K_{s})=R_{+}(K_{s})+\frac{1}{4}K_{0}\partial^{-1}\gamma_{0}^{T}K_{s}%
\]
and obviously $R_{+}(K_{s})$ is local due to the assumption that $K_{s}$ is
local. Since all $K_{s}$ are symmetries for all $\gamma_{s}$ we have that
$L_{K_{s}}\gamma_{0}=0$ which since $d\gamma_{0}=0$ yields $d\left\langle
\gamma_{0},K_{s}\right\rangle =0$ where $d$ is the operator of exterior
differentiation and where $\left\langle \cdot,\cdot\right\rangle $ is the dual
map between cotangent and tangent spaces. Thus, $\gamma_{0}^{T}K_{s}$ is a
total derivative so that $\partial^{-1}\gamma_{0}^{T}K_{s}$ is purely local.
This completes the inductive step. The proof for the negative part is similar.
Since $K_{-1}$ is local we have%
\[
K_{s-1}=R^{-1}(K_{s})=\left(  R^{-1}\right)  _{+}(K_{s})+K_{-1}\partial
^{-1}\gamma_{-1}^{T}K_{s}%
\]
and $\partial^{-1}\gamma_{-1}^{T}K_{s}$ is purely local by the same argument
as above, since $L_{K_{s}}\gamma_{-1}=0$.
\end{proof}

A similar theorem is valid for one-forms $\gamma_{s}$.

\begin{theorem}
The one-forms $\gamma_{s}$ are local for all $s\in\mathbf{Z}$.
\end{theorem}

\begin{proof}
The proof is analogous to the proof of previous theorem: one proves it by
induction with respect to $s$ separately for the positive and for the negative
hierarchy. We give the proof only for the positive hierarchy. The (nontrivial)
one-form $\gamma_{0}$ is local. Assume that $\gamma_{s}$ is local. Then, since
$R^{\dagger}=\left(  R^{\dagger}\right)  _{+}+\left(  R^{\dagger}\right)
_{-}=\left(  R^{\dagger}\right)  _{+}-\frac{1}{4}\gamma_{0}\partial^{-1}%
K_{0}^{T}$,%
\[
\gamma_{s+1}=\left(  R^{\dagger}\right)  _{+}(\gamma_{s})-\frac{1}{4}%
\gamma_{0}\partial^{-1}K_{0}^{T}\gamma_{s}.
\]
But obviously $\left(  R^{\dagger}\right)  _{+}(\gamma_{s})$ is local due to
the assumption while, since $L_{K_{0}}\gamma_{s}=0$, $K_{0}^{T}\gamma_{s}$ is
a total derivative so that $\partial^{-1}K_{0}^{T}\gamma_{s}$ is purely local.
This completes the induction.
\end{proof}

The situation is different when we consider the Poisson operators $B_{s}$.

\begin{theorem}
The operators $B_{0},\ldots,B_{N}$ are local. All the others operator $B_{s} $
are nonlocal of order $1$ with the nonlocal terms of the form%
\begin{align*}
\left(  B_{N+s}\right)  _{-}  & =-\frac{1}{4}%
{\displaystyle\sum\limits_{j=1}^{s}}
K_{j-1}\partial^{-1}K_{s-j+1}^{T}\text{, \ \ }s\in\mathbf{Z}_{+}\\
\left(  B_{-s}\right)  _{-}  & =-%
{\displaystyle\sum\limits_{j=1}^{s}}
K_{-j}\partial^{-1}K_{-s+j-1}^{T}\text{, \ \ }s\in\mathbf{Z}_{+}%
\end{align*}

\end{theorem}

\begin{proof}
Due to the fact that $R_{-}=\frac{1}{4}K_{0}\partial^{-1}\gamma_{0}^{T}$ we
have, according to formula (\ref{fin}) in Theorem \ref{potega}%
\begin{align}
\left(  R^{s}\right)  _{-}  & =\left(  R^{s}\right)  _{-1}=\frac{1}{4}%
{\displaystyle\sum\limits_{j=0}^{s-1}}
R^{j}(K_{0})\partial^{-1}\left[  \left(  R^{\dagger}\right)  ^{s-1-j}%
\gamma_{0}\right]  ^{T}=\frac{1}{4}%
{\displaystyle\sum\limits_{j=0}^{s-1}}
K_{j}\partial^{-1}\gamma_{s-1-j}^{T}=\label{rsm}\\
& =\frac{1}{4}%
{\displaystyle\sum\limits_{j=1}^{s}}
K_{j-1}\partial^{-1}\gamma_{s-j}^{T}\nonumber
\end{align}
for all $s\in\mathbf{Z}_{+}$. Thus, for all $s\in\mathbf{Z}_{+}$ and due to
the fact that $B_{0}$ is local we have $\left(  B_{s}\right)  _{-}=\left(
R^{s}B_{0}\right)  _{-}=\left(  R^{s}\right)  _{-}B_{0}$ which by formula
(\ref{rsm}) above yields%
\[
\left(  B_{s}\right)  _{-}=\frac{1}{4}%
{\displaystyle\sum\limits_{j=1}^{s}}
K_{j-1}\partial^{-1}\gamma_{s-j}^{T}B_{0}=\frac{1}{4}%
{\displaystyle\sum\limits_{j=1}^{s}}
K_{j-1}\partial^{-1}\left(  B_{0}^{\dagger}\gamma_{s-j}\right)  ^{T}=-\frac
{1}{4}%
{\displaystyle\sum\limits_{j=1}^{s}}
K_{j-1}\partial^{-1}\left(  B_{0}\gamma_{s-j}\right)  ^{T}%
\]
since $B_{0}^{\dagger}=-B_{0}$ and due to the fact that $\partial^{-1}%
\gamma^{T}B=\partial^{-1}(B^{\dagger}\gamma)^{T}$ for any linear
pseudo-differential operator $B$ and any column matrix $\gamma$ with entries
depending on $x$. Thus, since $B_{0}\gamma_{s-j}=0$ for $s=1,\ldots,N$ and
$j=1,\ldots,s$ ($\gamma_{s-j}$ are the Casimir forms for $B_{0}$ for all
$s=1,\ldots,N$ and $j=1,\ldots,s$ due to Lemma \ref{cas}) we see that $\left(
B_{s}\right)  _{-}=0$ for $s=0,\ldots,N$. Moreover, due to (\ref{26}),
$B_{0}\gamma_{N+s-j}=K_{s-j}$ for $j=1,\ldots,s$ and thus%
\[
\left(  B_{N+s}\right)  _{-}=-\frac{1}{4}%
{\displaystyle\sum\limits_{j=1}^{N+s}}
K_{j-1}\partial^{-1}\left(  B_{0}\gamma_{N+s-j}\right)  ^{T}=-\frac{1}{4}%
{\displaystyle\sum\limits_{j=1}^{s}}
K_{j-1}\partial^{-1}K_{s-j+1}^{T}.
\]
Similarly, due to the fact that $\left(  R^{-1}\right)  _{-}=K_{-1}%
\partial^{-1}\gamma_{-1}^{T}$ we obtain, by the same formula (\ref{fin}) in
Theorem \ref{potega} that for all $s\in\mathbf{Z}_{+}$%
\begin{align}
\left(  R^{-s}\right)  _{-}  & =\left(  R^{-s}\right)  _{-1}=%
{\displaystyle\sum\limits_{j=0}^{s-1}}
R^{-j}(K_{-1})\partial^{-1}\left[  \left(  \left(  R^{-1}\right)
\dagger\right)  ^{s-1-j}\left(  \gamma_{-1}\right)  \right]  ^{T}=%
{\displaystyle\sum\limits_{j=0}^{s-1}}
K_{-j-1}\partial^{-1}\gamma_{-s+j}^{T}=\label{rmsm}\\
& =%
{\displaystyle\sum\limits_{j=1}^{s}}
K_{-j}\partial^{-1}\gamma_{-s+j-1}^{T}\nonumber
\end{align}
and thus by computations similar to the above we obtain%
\begin{align*}
\left(  B_{-s}\right)  _{-}  & =\left(  R^{-s}\right)  _{-}B_{0}=%
{\displaystyle\sum\limits_{j=1}^{s}}
K_{-j}\partial^{-1}\gamma_{-s+j-1}^{T}B_{0}=%
{\displaystyle\sum\limits_{j=1}^{s}}
K_{-j}\partial^{-1}\left(  B_{0}^{\dagger}\gamma_{-s+j-1}\right)  ^{T}=\\
& =-%
{\displaystyle\sum\limits_{j=1}^{s}}
K_{-j}\partial^{-1}\left(  B_{0}\gamma_{-s+j-1}\right)  ^{T}=-%
{\displaystyle\sum\limits_{j=1}^{s}}
K_{-j}\partial^{-1}K_{-s+j-1}^{T}.
\end{align*}
Thus, all $B_{s}$ except $B_{0},\ldots,B_{N}$ are nonlocal of order $1$.
\end{proof}

Finally, let us discuss the nonlocalities of symplectic forms $\Omega_{s}$
defined through (\ref{omegi}). In order to establish the leading nonlocal term
of $\Omega_{s}$ by applying Theorems \ref{zlozenie} and \ref{potega} to
$\Omega_{s}=\left(  R^{\dagger}\right)  ^{s}\Omega_{0}$ we need first to
establish the leading nonlocal term of $\Omega_{0}.$

\begin{theorem}
$\Omega_{0}$ has nonlocality of first order and its leading nonlocal term is
given by%
\begin{equation}
\left(  \Omega_{0}\right)  _{-}=\frac{1}{4}%
{\displaystyle\sum\limits_{j=0}^{N-1}}
\gamma_{j}\partial^{-1}\gamma_{N-1-j}^{T}.\label{Ominus}%
\end{equation}

\end{theorem}

One can prove this formula by showing that (\ref{Ominus}) is a simultaneous
solution to all equations%
\[
\left(  B_{s}\Omega_{0}\right)  _{-}=\left(  R^{s}\right)  _{-}\text{, \ }%
s\in\mathbf{Z.}%
\]
We skip the proof as the computations involved are similar to those for
nonlocal parts of $B_{s}$ operators.

Now we can calculate the nonlocalities of $\Omega_{s}$ by applying Theorems
\ref{zlozenie} and \ref{potega} to $\Omega_{s}=\left(  R^{\dagger}\right)
^{s}\Omega_{0}$. The result of this calculation is presented (without proof)
in the theorem below.

\begin{theorem}
The closed two-forms $\Omega_{s}=\left(  R^{\dagger}\right)  ^{s}\Omega_{0}$
are nonlocal of first order (weakly nonlocal) with nonlocal terms given by%
\begin{align*}
\left(  \Omega_{s}\right)  _{-}  & =\frac{1}{4}%
{\displaystyle\sum\limits_{j=0}^{N+s-1}}
\gamma_{j}\partial^{-1}\gamma_{s+N-1-j}^{T}\text{, \ \ }s\in\mathbf{Z}_{+},\\
\left(  \Omega_{-s}\right)  _{-}  & =\frac{1}{4}%
{\displaystyle\sum\limits_{j=0}^{N-s-1}}
\gamma_{j}\partial^{-1}\gamma_{-s+N-1-j}^{T}+%
{\displaystyle\sum\limits_{j=1}^{s}}
\gamma_{-j}\partial^{-1}\gamma_{-s+j-1}^{T}\text{, \ \ }s\in\mathbf{Z}_{+}.
\end{align*}

\end{theorem}

\section{Invertible coupled Harry Dym hierarchy}

There is a second possibility for choosing the constants in $J\,_{0}$ and
$J_{N}$ so that both operators become invertible, which guarantees the
existence of both positive and negative hierarchies, namely, to take, as
before $\varepsilon_{N}=0$ and then to put $u_{0}=0$. The operators $J_{i}$
attain then the form%
\begin{align*}
J_{0}  & =\frac{1}{4}\varepsilon_{0}\partial^{3},\\
J_{i}  & =\frac{1}{4}\varepsilon_{i}\partial^{3}+\frac{1}{2}\left(
u_{i}\partial+\partial u_{i}\right)  ,\text{ \ }i=1,\ldots,N-1,\\
J_{N}  & =\frac{1}{2}(u_{N}\partial+\partial u_{N}).
\end{align*}
We can therefore define the invertible coupled Harry Dym (invertible cHD)
hierarchy as follows:

\begin{definition}
The invertible $N$-component coupled Harry Dym hierarchy is the family of
flows%
\begin{align}
\frac{d}{dt_{-s}}u  & =K_{-s}\equiv B_{r}\gamma_{-r-s}\text{, \ }%
s=1,2,\ldots\text{ and for all }r>-s\nonumber\\
& \\
\frac{d}{dt_{s}}u  & =K_{s}\equiv B_{r}\gamma_{-r+s+N}\text{, \ }%
s=0,1,2,\ldots\text{ and for all }r\leq s+N\nonumber
\end{align}
where $u=(u_{1},\ldots,u_{N})^{T}$, $u_{i}=u_{i}(x,t)$, $B_{r}$ are
hamiltonian operators defined in (\ref{Br}) with $\varepsilon_{N}=0$,
$u_{0}=0$ and where $\gamma_{s}$ are one-forms given by (\ref{gammy})%
\begin{equation}
\gamma_{s}=\left(  R^{\dagger}\right)  ^{s}\gamma_{0},\text{ \ \ \ }%
\gamma_{-s}=\left(  \left(  R^{-1}\right)  ^{\dagger}\right)  ^{s-1}%
\gamma_{-1},\text{ \ \ }s=1,2,\ldots
\end{equation}
with
\begin{equation}
\gamma_{0}=\left(  0,\ldots,0,u_{N}^{-1/2}\right)  ^{T}\in\text{Ker}\left(
(R^{\dagger})^{-1}\right)  ,\text{ \ \ \ }\gamma_{-1}=(-2,0,\ldots,0)^{T}%
\in\text{Ker}(R^{\dagger})\label{wybghd}%
\end{equation}
(again $\gamma_{0}$ and $\gamma_{-1}$ are chosen as in (\ref{start})). Thus,
formally, both hierarchies have the same algebraic structure (apart from the
fact that now the nontrivial variables are $u_{1},\ldots,u_{N}$); the formulas
(\ref{26})-(\ref{28}) are still valid, but of course the exact form of
invariant one-forms $\gamma_{i}$ and vector fields $K_{i}$ differ. Explicitly,
for the invertible cHD hierarchy we have%
\begin{align*}
\left(  K_{-2}\right)  _{i}  & =u_{i+1,x}-\frac{1}{\varepsilon_{0}}%
\varepsilon_{i}u_{1,x}-\frac{4}{\varepsilon_{0}}u_{i}\partial^{-1}u_{1}%
-\frac{2}{\varepsilon_{0}}u_{i,x}\partial^{-2}u_{1}\\
\left(  K_{-1}\right)  _{i}  & =u_{i,x}\\
\left(  K_{0}\right)  _{i}  & =\frac{1}{4}\varepsilon_{i-1}\left(
u_{N}^{-1/2}\right)  _{xxx}+u_{i-1}\left(  u_{N}^{-1/2}\right)  _{x}+\frac
{1}{2}u_{i-1,x}u_{N}^{-1/2}%
\end{align*}
(with $u_{i}=0$ for $i<1$ or $i>N$) so that the negative part is no longer
local. Below we prove that the positive part is still local. Note also that
this hierarchy depends now on $N$ parameters: $\varepsilon_{0},\ldots
,\varepsilon_{N-1}$. The first few conserved one-forms of the invertible cHD
hierarchy are%
\begin{align*}
\gamma_{-2}  & =\left(  \frac{4}{\varepsilon_{0}}\partial^{-2}u_{1}%
,-2,0,\ldots,0\right)  ^{T}\\
\gamma_{-1}  & =(-2,0,\ldots,0)^{T}\\
\gamma_{0}  & =\left(  0,\ldots,0,u_{N}^{-1/2}\right)  ^{T}\\
\gamma_{1}  & =\left(  0,\ldots,0,u_{N}^{-1/2},-\frac{1}{4}\varepsilon
_{N-1}\left[  u_{N}^{-1}\left(  u_{N}^{-1/2}\right)  _{xx}-\frac{1}{8}%
u_{N}^{-7/2}\left(  u_{N,x}\right)  ^{2}\right]  -\frac{1}{2}u_{N}%
^{-3/2}u_{N-1}\right)  ^{T}.
\end{align*}

\end{definition}

As in the cKdV case, the positive part of the above hierarchy coincides with
the coupled Harry Dym hierarchy in \cite{AF4} after putting $a=0$. We also see
that the negative part of the sequence of $\gamma_{i}$ is nonlocal; we prove
below that the positive $\gamma_{i}$'s are local. Let us thus consider the
leading nonlocal parts of all the objects of the invertible cHD hierarchy. We
first establish the nonlocal parts of the recursion operator $R $ and its
inverse $R^{-1}$. We obtain (cf.\ (\ref{split}))%
\begin{align}
R  & =R_{+}+R_{-}=R_{+}-K_{0}\partial^{-1}\gamma_{0}^{T}\label{split2}\\
R^{-1}  & =\left(  R^{-1}\right)  _{+}+\left(  R^{-1}\right)  _{-}=\left(
R^{-1}\right)  _{+}+\frac{2}{\varepsilon_{0}}u\partial^{-2}\gamma_{-1}%
^{T}+\frac{1}{\varepsilon_{0}}K_{-1}\partial^{-3}\gamma_{-1}^{T}.\nonumber
\end{align}
Thus, contrary to the cKdV case, $R^{-1}$ has nonlocality of order $3$ while
$R$ is still nonlocal of order $1$. Similarly to the cKdV case, we obtain by
induction that

\begin{theorem}
The vector fields $K_{i}$ and the conserved one-forms $\gamma_{i}$ are local
for all $i\in\mathbf{Z}_{+}$.
\end{theorem}

The structure of nonlocalities in the operators $B_{s}$ differs from the cKdV case.

\begin{theorem}
The operators $B_{0},\ldots,B_{N}$ are local. All the others operator $B_{s} $
are nonlocal with the leading nonlocal terms of the form%
\begin{align*}
\left(  B_{s+N}\right)  _{-1}  & =%
{\displaystyle\sum\limits_{j=1}^{s}}
K_{j-1}\partial^{-1}K_{s-j+1}^{T}\text{, \ \ }s\in\mathbf{Z}_{+}\\
\left(  B_{-s}\right)  _{-3}  & =-\frac{1}{\varepsilon_{0}}%
{\displaystyle\sum\limits_{j=0}^{s-1}}
K_{-j-1}\partial^{-3}K_{-s+j}^{T}\text{, \ \ }s\in\mathbf{Z}_{+}%
\end{align*}

\end{theorem}

Again, one proves this theorem by applying Theorem \ref{zlozenie} and Theorem
\ref{potega} to the operator $B_{s}=R^{s}B_{0}$. Thus, the negative operators
are nonlocal of order $3$. The above theorem can now be used to discuss the
nonlocalities of the closed two-forms $\Omega_{s}$ defined through
(\ref{omegi}). We begin by establishing the form of the nonlocal part of
$\Omega_{0}$.

\begin{theorem}
$\Omega_{0}$ has nonlocality of first order and its leading nonlocal term is
given by%
\begin{equation}
\left(  \Omega_{0}\right)  _{-}=-%
{\displaystyle\sum\limits_{j=0}^{N-1}}
\gamma_{j}\partial^{-1}\gamma_{N-1-j}^{T}.\label{ominus2}%
\end{equation}

\end{theorem}

This theorem can be proved just as in the cKdV case. Now, by applying again
Theorem \ref{zlozenie} and Theorem \ref{potega} to the operator $\Omega
_{s}=\left(  R^{\dagger}\right)  ^{s}\Omega_{0}$ with $\Omega_{0}$ given by
(\ref{ominus2}) we obtain the following theorem.

\begin{theorem}
The symplectic operators $\Omega_{s}=\left(  R^{\dagger}\right)  ^{s}%
\Omega_{0}$ are nonlocal of first order for $s>0$ and of third order for $s<0$
and their leading nonlocal terms are given by%
\begin{align*}
\left(  \Omega_{s}\right)  _{-}  & =-%
{\displaystyle\sum\limits_{j=0}^{N+s-1}}
\gamma_{j}\partial^{-1}\gamma_{s+N-1-j}^{T}\text{, \ \ }s\in\mathbf{Z}_{+}\\
\left(  \Omega_{-s}\right)  _{-3}  & =-\frac{1}{\varepsilon_{0}}%
{\displaystyle\sum\limits_{j=1}^{s}}
\gamma_{-j}\partial^{-3}\gamma_{-s+j-1}^{T}\text{, \ \ }s\in\mathbf{Z}_{+}%
\end{align*}

\end{theorem}

\section{Conclusions}

In this paper we have stated some natural conditions under which the coupled
Korteweg-de Vries and the coupled Harry Dym hierarchies, introduced in the
paper \cite{JM} in the two-field case and developed in papers \cite{MA}%
-\cite{AF5}, are invertible, i.e., possess negative parts. We studied the
structure of nonlocalities of various tensor invariants of these hierarchies.
It turns out that all the vector fields and conserved one-forms of the
invertible cKdV hierarchy are local while all its Poisson operators are either
local or at most weakly nonlocal. Finally, all symplectic operators of the
hierarchy are weakly nonlocal. In the case of the invertible cHD hierarchy
only vector fields and conserved one-forms of the positive part are local,
while Poisson operators of this hierarchy are either local or have
nonloclities of first or third order. Moreover, all symplectic operators of
the invertible cHD hierarchy are nonlocal of first or third order. The main
tool for our considerations was a generalization of formulae (7) and (8) from
\cite{S}, that is, Theorems \ref{zlozenie} and \ref{potega} above.

\section{Acknowledgement}

This paper was partially financed by Swedish Research Council grant no. 2011-52.

\end{document}